\documentclass{article}
\usepackage[numbers,sort&compress]{natbib}
\makeatletter \renewcommand*{\NAT@spacechar}{~}
\makeatother

\author{Szymon Grabowski$^1$ \and Dominik K\"{o}ppl$^2$}
\date{{\small
  $^1$Lodz University of Technology, Institute of Applied Computer Science, 
Poland\\$^2$M\&D Data Science Center, Tokyo Medical and Dental University, Japan
  }
}

\usepackage[utf8]{inputenc}

\usepackage{hyperref}
\usepackage{amssymb}
\usepackage{amsthm} \usepackage[capitalise]{cleveref} 

\theoremstyle{definition}
\newtheorem{theorem}{Theorem}[section]

\newtheorem{lemma}[theorem]{Lemma}

\newtheorem{corollary}[theorem]{Corollary}

\newcommand{\bsq}[1]{\lq{#1}\rq} \newcommand{\gauss}[1]{\left\lfloor#1\right\rfloor} \newcommand{\upgauss}[1]{\left\lceil#1\right\rceil} 

\newcommand{\UnaryOperator}[2][]{\ifx&#1&\ensuremath{\mathop{}\mathopen{}#2\mathopen{}}\else \ensuremath{\mathop{}\mathopen{}#2\mathopen{}(#1)}\fi }
\newcommand{\Oh}[1]{\UnaryOperator[#1]{\mathcal{O}}}
\newcommand{\oh}[1]{\UnaryOperator[#1]{o}}
\newcommand{\om}[1]{\UnaryOperator[#1]{\mathup{\omega}}}
\newcommand{\Om}[1]{\UnaryOperator[#1]{\mathup{\Omega}}}

\newcommand{\Ot}[1]{\UnaryOperator[#1]{\mathup{\Theta}}}
\DeclareMathAlphabet{\mathup}{OT1}{msb}{m}{n}

\newcommand{\maxLength}{\ensuremath{\ell_{\mathup{max}}}}
\newcommand*{\Codewords}{\ensuremath{\mathcal{C}}}
\newcommand*{\predecessor}{\mathop{pred}}
\newcommand*{\rank}{\mathop{rank}}
\newcommand*{\select}{\mathop{select}}
\newcommand*{\First}{\ensuremath{\mathsf{First}}}
\newcommand*{\Fusion}{\ensuremath{\mathsf{P}}}
\newcommand*{\BV}[1]{\ensuremath{\mathsf{B}_{\mathup{#1}}}}

\title{Space-Efficient Huffman Codes Revisited}

\begin{document}

\maketitle

\begin{abstract}
Canonical Huffman code is an optimal prefix-free compression code whose codewords enumerated in the lexicographical order form a list of binary words in non-decreasing lengths.
Gagie et al. (2015) gave a representation of this coding capable
to encode or decode a symbol in constant worst case time.
It uses $\sigma \lg \maxLength + \oh{\sigma}  + \Oh{\maxLength^2}$ bits of space, 
where $\sigma$ and $\maxLength$ are the alphabet size and maximum codeword length, respectively.
We refine their representation to reduce the space complexity to $\sigma \lg \maxLength (1 + \oh{1})$ bits while preserving the constant encode and decode times.
Our algorithmic idea can be applied to any canonical code.
\end{abstract}

\section{Introduction} \label{sec:intro}
Huffman coding~\cite{Huffman52,Moffat19}, 
whose compressed output closely approaches the zeroth-order empirical entropy,
is nowadays perceived as a standard textbook encoder, being the most prevalent option when it comes to statistical compression.
Unlike arithmetic coding~\cite{witten87arithmetic}, its produced code is \emph{instantaneous}, meaning that a character can be decoded as soon as the last bit of its codeword is read from a compressed input stream.
Thanks to this property, instantaneous codes tend to be swiftly decodable.
While most research efforts focus on the achievable decompression speeds, not much has been considered for actually representing the dictionary of codewords, which is usually done with a binary code tree, called the Huffman tree. 

Given that the characters of our text are drawn from an alphabet~$\Sigma$ of size~$\sigma$,
the \emph{Huffman tree} is a full binary tree, where each leaf represents a character of~$\Sigma$.
The tree stored in a plain pointer-based representation takes $\Oh{\sigma \lg \sigma}$ bits of space.
Here we assume that $\sigma \le n$ and $\Sigma = \{1, 2, \ldots, \sigma\}$, where $n$ is the number of characters of our text. A naive encoding algorithm for a character $c$ traverses the Huffman tree top-down to the
leaf representing $c$ while writing the bits stored as labels on the
traversed edges to the output. Similarly, a decoding algorithm traverses the tree top-down
based on the read bits until reaching a leaf. Hence, the decoding and
encoding need $\Oh{\maxLength}$ time, where $\maxLength$ is the height of
the tree.
If $\sigma$ is constant, then this representation is optimal with
respect to space and time (since $\maxLength$ becomes constant).
Therefore, the problem we study becomes interesting if $\sigma = \omega(1)$,
which is the setting of this article. 

The alphabet size $\sigma$ is often much smaller than $n$. Yet in some cases the Huffman tree space becomes non-negligible:
For instance, let us consider that $\Sigma$ represents a set of distinct words of a language.
Then its size can 
even exceed a 
million\footnote{Two such examples are the Polish Scrabble dictionary with over 3.0M words (\url{https://sjp.pl/slownik/growy/}) and the Korean dictionary \emph{Woori Mal Saem} with over 1.1M words (\url{https://opendict.korean.go.kr/service/dicStat}).}
for covering an entire natural language. 
In another case, we maintain a collection of texts compressed with the same Huffman codewords under the setting that a single text needs to be shipped together with the codeword dictionary, for instance across a network infrastructure.
A particular variation are \emph{canonical Huffman codes}~\cite{SK64}, 
where leaves of the same depth in the Huffman tree are lexicographically ordered with respect to the characters they represent.
An advantage of canonical Huffman codes is that 
several techniques (see~\cite{liddell06decoding,moffat97implementation} and the references therein) 
can compute the lengths of a codeword from the compressed stream by reading bit chunks instead of single bits,
thus making it possible to decode a character faster than an approach linear in the bit-length of its codeword.
Moffat and Turpin \cite[Algorithm TABLE-LOOKUP]{moffat97implementation} were probably the first to notice that it is enough to (essentially) perform binary search over a collection of the lexicographically smallest codewords for each possible length, to decode the current symbol.
The same idea is presented perhaps more lucidly in \cite[Section~2.6.3]{navarro16compact}, with an explicit claim that this representation requires $\sigma \lg \sigma + \Oh{\lg^2 n}$ bits\footnote{By $\lg$ we mean the logarithm to base two ($\log_2$) throughout the paper. Occasional logarithms to another base $b$ are denoted with $\log_b$.} and achieves \Oh{\lg \lg n} 
time per codeword. 
Given $\maxLength$ is the maximum length of the codewords,
an improvement, both in space and time, has been achieved by
\citet{gagie15efficient},
who gave a representation of the canonical Huffman tree within 
\begin{itemize}
   \item    $\sigma \lg \maxLength + \oh{\sigma}  + \Oh{\maxLength^2}$ bits of space with $w \in \Om{\maxLength}$ \cite[Theorem 1]{gagie15efficient}, or
   \item    $\sigma \lg \lg (n/\sigma) + \Oh{\sigma} + \Oh{\maxLength^2}$ bits of space with $w \in \Om{\lg n}$ \cite[Corollary 1]{gagie15efficient},\footnote{Although the space is stated as $\sigma \lg \lg (n/\sigma) + \Oh{\sigma} + \Oh{\lg^2 n}$ bits in that paper, we could verify with the authors that the last term can be improved to the bounds written here.}
\end{itemize}
while supporting character encoding and decoding in \Oh{1} time.
Their approach consists of a wavelet tree, a predecessor data structure~$\Fusion$, and an array~$\First$ storing the lexicographically smallest codeword for each depth of the Huffman tree.
The last two data structures take account for the last term in each of the space complexities above.

\paragraph{Our Contribution}

Our contribution is a joint representation of the predecessor data structure $\Fusion$ with the array $\First$ 
to improve their required space of $\Oh{\maxLength^2}$ bits down to $\oh{\maxLength \lg \sigma + \sigma}$ bits,
and this space is in $\oh{\sigma \lg \maxLength}$ bits since $\maxLength \leq \sigma - 1$.
In other words, we obtain the following theorem.

\begin{theorem}\label{thmOverallResult}
There is a data structure using $\sigma \lg \maxLength (1 + \oh{1})$ bits of
space,
which can encode a character to a canonical Huffman codeword or restore a
character from a binary input stream of codewords, both in constant time per
character.
\end{theorem}

We remark that when the alphabet consists of natural language words, like in the examples shown in the beginning, we often initially map words to IDs in an arbitrary manner.
For our use case, it would be convenient to assign the IDs based on their frequencies.
By doing so, the needed space of canonical Huffman representation is greatly reduced as the alphabet permutation can be considered as already given. 
In such a case, we only consider the extra space needed for the Huffman coding, where our improvement in space from $\Oh{\maxLength^2}$ down to $\oh{\maxLength \log \sigma}$ is more pronounced (in detail, we can omit the wavelet tree of \citeauthor{gagie15efficient}'s achieved space complexities).

\paragraph{Related Work}
A lot of research effort has been spent on the practicality of decoding Huffman-encoded texts (e.g.,~\cite{hirschberg90efficient,aggarwal00efficient,moffat97implementation,hashemian95memory,mansour07efficient})
where often a lookup table is employed for figuring out the lengths of the currently read codeword.
For instance, \citet{nekrich00decoding} added a lookup table for byte chunks of a codeword to figure out the length of a codeword by reading it byte-wise instead of bit-wise.
Regarding different representations of the Huffman tree,
Chowdhury et al.~\cite{chowdhury02efficient} came up with a solution using  $\upgauss{3\sigma/2}$ words of space,
which can decode a codeword in \Oh{\log \sigma} time.

In theory terms, important results were given by Gagie et al.~\cite{gagie15efficient} and Fari{\~n}a et al.~\cite{FGMNOspire16} (see also an extended version~\cite{farina2021efficient}).
The former is heavily cited in this work, and is discussed in detail in~\cref{secFormerApproach}.
In the latter, it is shown how to represent a particular optimal prefix-free code used for compressing wavelet matrices~\cite{CNOis14} in $\Oh{\sigma \lg\maxLength + 2^{\varepsilon \maxLength}}$ bits, allowing $\Oh{1/\varepsilon}$-time encoding and decoding, for a selectable constant $\varepsilon > 0$.
Like \citet{gagie15efficient}, they also use a wavelet tree of \Oh{\sigma \lg\maxLength} bits to map each character to the length of its respective codeword.
The tree topology is represented by counting level-wise the number of nodes and leaves, resulting in $\Oh{\maxLength \lg \sigma} \subset  \Oh{\sigma \lg \maxLength}$ bits.
With these two ingredients, this structure is already operational with \Oh{\maxLength} time per encoding or decoding operation.
To obtain the aforementioned time bounds, they sample certain depths of the code tree with lookup tables to speed-up top-down traversals.

Another line of research is on so-called skeleton trees~\cite{klein97space,klein00skeleton,kosolobov20optimal}.
The idea is to replace disjoint perfect subtrees in the canonical Huffman tree with leaf nodes.
A leaf node representing a pruned perfect subtree only needs to store the height of this tree and the characters represented by its leaves to be able to reconstruct it. 
Note that all leaves in a perfect binary tree are on the same level, and hence due to the restriction on the order of the leaves of the canonical Huffman tree, we can restore the pruned subtree by knowing its height and the set of characters.
Thus, a skeleton tree may use less space. 
Decompression can be accelerated whenever we hit a leaf~$\lambda$ during a top-down traversal, where we know that the next $h$ bits from the input are equal to the suffix of the currently parsed codeword if the leaf~$\lambda$ is the representative of a perfect subtree of height~$h$.
Unfortunately, the gain depends on the shape of the Huffman tree, and since there are no non-trivial theoretical worst-case bounds known, 
this approach can be understood as a heuristic so far.

\section{Preliminaries}
Our computational model is the standard word RAM with machine word size~$w = \Om{\lg n}$ bits,
where $n$ denotes the length of a given input string~$T[1..n]$ (called the \emph{text}) 
whose characters are drawn from an integer alphabet~$\Sigma = \{1, 2, \ldots, \sigma\}$ with size~$\sigma$ 
being bounded by $\sigma \le n$ and $\sigma = \om{1}$. 
We call the elements of $\Sigma$ \emph{characters}.
Given a string $S \in \Sigma^*$, 
we define the queries $S.\rank_c(i)$ and $S.\select_c(j)$ returning the number of $c$'s in $S[1..i]$ and the position of the $j$-th $c$ in $S$, respectively.
There are data structures that replace $S$ and use $|S|\lg \sigma + \oh{|S|}$ bits, and can answer each query in
\Oh{1 + \lg \log_w \sigma} time~\cite[Theorem 4.1]{belazzougui15wavelet}.
We further stipulate that $S.\select_c(0) := 0$ for any~$c \in \Sigma$.
In what follows, we assume that we have at least \Ot{w} bits of working space available for storing a constant number of pointers, and omit this space in the space analysis of this article.

\subsection{Code Tree}
A \emph{code tree} is a full binary tree whose leaves represent the characters of $\Sigma$.
A left or a right child of a node~$v$ is connected to~$v$ via an edge with label \texttt{0} or \texttt{1}, respectively.
We say that the \emph{string label} of a node~$v$ is 
the concatenation of edge labels on the path from the root to~$v$. 
Then the string label of a leaf~$\lambda$ is the codeword of the character represented by~$\lambda$.
The length of the string label of a node is equal to its depth in the tree.
Let $\Codewords$ denote the set of the string labels of all leaves.
The bit strings of $\Codewords$ are called \emph{codewords}.
A crucial property of a code tree is that $\Codewords$ is prefix-free, meaning that no codeword in $\Codewords$ is prefix of another.
A consequence is that, when reading from a binary stream of concatenated codewords, 
we only have to match the longest prefix of this stream with $\Codewords$ to decode the next codeword.
Hence, prefix-free codes are instantaneous, where an \emph{instantaneous code}, as described in the introduction, is a code with the property that a character can be decoded as soon as the last bit of its codeword is read from the compressed input.
Instantaneous codes are also prefix-free.
That is because, if a code is not prefix-free, then there can be two codewords $C_1$ and $C_2$ with $C_2$ being a prefix of $C_1$ such that when reading $C_2$ we need to read the next $|C_1|-|C_2|$ bits to judge whether the read bits represent $C_1$ or $C_2$.

A code tree is called \emph{canonical}~\cite{SK64} if its induced codewords
read from left to right are lexicographically strictly increasing, while 
their lengths are non-decreasing.
Consequently, a codeword of a shallow leaf is both shorter and
lexicographically smaller than the codeword of a deeper leaf.
We further assume that all leaves on the
same depth are sorted according to the order 
of their respective characters. 

Let $\maxLength$ denote the  maximum length of the codewords in \Codewords{}.
In the following we fix an instance of a code tree
for which $\maxLength{} \le \min(\sigma-1, \log_\phi n)$ holds, where
$\phi = (1 + \sqrt{5})/2 \approx 1.618$ is the golden ratio. 
It has been shown by \citet[comment after Theorem 2.1]{buro93maximum} that the canonical Huffman coding exhibits this property.

Our claims in \cref{sec:results,sec:warmup}, expressed in terms of the canonical Huffman codes and trees, are applicable to arbitrary canonical codes (resp.\ trees) if they obey the above bound on $\maxLength$.
Our focus on Huffman codes is motivated by the importance of this coding.

\subsection{Former Approach}\label{secFormerApproach}
We briefly review the approach of \citet{gagie15efficient},
which consists of three data structures: 
\begin{enumerate}
   \item a multi-ary wavelet tree storing, for each character, the length of its corresponding codeword,
   \item an array \First{} storing the codeword of the leftmost leaf on each depth, and
   \item a predecessor data structure~\Fusion{} for finding the length of a codeword that is the prefix of a read bit string.
\end{enumerate}
The first data structure represents an array~$L[1..\sigma]$ with $L[i]$ being the codeword length of the $i$-th character, i.e.,
$L$ is a string of length~$\sigma$ whose characters are drawn from the alphabet $\{1, 2, \ldots, \maxLength\}$. 
The wavelet tree built upon~$L$ can access $L$, and answer $L.\rank$ and $L.\select$ in \Oh{1 + \log_w \maxLength} time~\cite[Theorem 4.1]{belazzougui15wavelet}.
While a plain representation costs $\sigma \lg \maxLength + \oh{\sigma}$ bits,
the authors showed that they can use an empirical entropy-aware topology of the wavelet tree to get the space down to $\sigma \lg \lg (n/\sigma)$ bits,
still accessing~$L$, and answering $\rank$ and $\select$ in $\Oh{1 + \log_w \maxLength}$ time~\cite[Theorem 5.1]{belazzougui15wavelet}.
Since $\maxLength \in \Oh{\log n}$, $\Oh{\log_w \maxLength} \subset \Oh{1}$, and hence the time complexity is constant.

The second data structure~$\First$ is a plain array of length $\maxLength$. 
Each of its entries has $\maxLength$ bits, therefore it takes \Oh{\maxLength^2} bits in total.
It additionally stores $\maxLength$ in $\Oh{\lg \maxLength}$ bits (for instance, in a prefix-free code such as Elias-$\gamma$~\cite{elias75universal}).

The last data structure~$\Fusion$ is represented by a fusion tree~\cite{fredman93fusion} storing for each depth the lexicographically smallest codeword padded to $\maxLength$ bits at its right end, where we assume the least significant bit is stored. 
A fusion tree storing $m$ elements of a set~$\mathcal{S}$, each represented in $m$ bits, needs \Oh{m^2} bits of space and can answer the query 
$\predecessor(X)$ returning the predecessor of $X$ in $\mathcal{S}$, i.e., $\max \{ Y \in \mathcal{S} : Y \le X \}$
in $\Oh{\log_w m}$ time. 
For our application, the fusion tree~$\Fusion$ built on codewords computes $\predecessor$ in \Oh{\log_w \maxLength} time, and needs \Oh{\maxLength^2} bits.
Like above, the query time is constant since $\maxLength \in \Oh{\log n} \subset \Oh{w}$.
Here, we want $\Fusion$ to return the depth~$d$ of the returned codeword, not the codeword itself (which we could retrieve by $\First[d]$). 
Unfortunately, it is not explicitly mentioned by \citet{gagie15efficient} how to obtain this depth,
in particular when some depths may have no leaves.
We address this problem with the following subsection and explain after that how the actual computation is done (\cref{secQueries}).

\subsubsection{Missing Leaves}\label{secMissingLeaves}
To extract the depth from a predecessor query on \Fusion{},
we replace each right-padded codeword~$C_i$ by the pair~$(C_i, \ell_i)$ as the keys stored in \Fusion{}, where
$\ell_i$ is the length of $C_i$.
Such a pair of codeword and length is represented by the concatenation of the binary representations of its two components
such that we can interpret this  pair as a bit string of length~$\maxLength + \lg \maxLength$.
By slightly abusing the notation, we can now query 
$\predecessor((B, \ell))$ to obtain $(\bar{B}, \bar{\ell})$, 
i.e., the argument for $\predecessor(\cdot)$ is a pair rather than a single value and similarly the returned value is a pair, but physically the bit string $\bar{B}\bar{\ell}$ is a predecessor of $B\ell$.
Then $\predecessor((X,\mathtt{1}^{\lg \maxLength})) = (C_i,\ell_i)$ for a bit string~$X \in \{0,1\}^{\maxLength}$ gives us not only the predecessor $C_i$ of $X$, but also $\ell_i$. 
Storing such long bit strings poses no problem as the time complexity of $\predecessor(X)$ for bit string $X$ does not change in a fusion tree when the length of $X$ is, e.g., doubled.
Appendix~\ref{secAugmentation} outlines a different strategy augmenting the fusion tree with additional methods instead of storing the codeword lengths.

\subsubsection{Encoding and Decoding}\label{secQueries}
Finally, we explain how the operational functionality is implemented, 
i.e., the steps of how to encode a character, and how to decode a character from a binary stream.
We can encode a character $c \in \Sigma$ by first finding the depth of the leaf~$\lambda$ representing $c$ with $\ell = L[c]$,
which is also the length of the codeword to compute.
Given the string label of the leftmost leaf~$\lambda'$ on depth~$\ell$ is $\First[\ell]$, then all we have to do is increment the value of this codeword by the number of leaves between~$\lambda$ and $\lambda'$ on the same depth.
For that we compute $L.\rank_\ell(c)$ 
that gives us the number of leaves to the left of~$\lambda$ on the same depth~$\ell$.
Hence, $\First[\ell] + L.rank_\ell(c) - 1$ is the codeword of~$c$.
For decoding a character, we assume to have a binary stream of concatenated codewords as an input. 
We first use the predecessor data structure to figure out the length of the first codeword in the stream.
To this end, we peek the next $\maxLength$ bits in the binary stream, i.e., we read $\maxLength$ bits into a variable~$X$ from the stream \emph{without} removing them.
Given $(C', \ell) = \Fusion.\predecessor((X,\mathtt{1}^{\lg \maxLength}))$,
we know that $\First[\ell] = C'$ and that the next codeword has length~$\ell$.
Hence, we can read $\ell$ bits from the stream into a bit string~$C$, which must be the codeword of the next character.
The leaf having $C$ as its string label is on depth~$\ell$, and has the rank $r = C - \First[\ell] + 1$ among all other leaves on the same depth.
This rank helps us to retrieve the character having the codeword~$C$,
which we can find by $L.\select_\ell(r)$.

\section{A Warm-Up}\label{sec:warmup}

We start with a simple idea unrelated to the main contribution presented in the subsequent section.
The point is that it might have a (mild) practical impact.
Also, we admit this idea is not new (Fari{\~n}a et al.~\cite[Section 2.4]{farina2021efficient} attribute it to Moffat and Turpin~\cite{moffat97implementation}, although, in our opinion, it is presented rather in disguise), we however hope that the analysis given below is original and of
some value.

Navarro, in Section~2.6.3 of his textbook \cite{navarro16compact}, describes a simple solution (based on an earlier work of Moffat and Turpin \cite{moffat97implementation}) which gives \Oh{\lg \lg n} worst-case time for symbol decoding.
With the modification presented below, this worst-case time remains, but the average time bound becomes \Oh{\lg \lg \sigma}. 
It also means we can decode the whole text in \Oh{n \lg \lg \sigma} rather than \Oh{n \lg \lg n} time.

\newcommand*{\Key}{\mathsf{key}}
\newcommand*{\Pos}{\mathit{pos}}

The referenced solution is based on a binary search over a collection of the lexicographically smallest codewords for each possible length; the collection is ordered ascendingly by the codeword length.
We replace the binary search with the exponential search~\cite{BY76}, which finds the predecessor of item $\Key$ in a (random access) list $L$ of $m$ items in \Oh{\lg \Pos_L(\Key)} time, where $\Pos_L(\Key)$ is the position of $\Key$ in $L$. 
Exponential search is not faster than binary search in general, but may help if $\Key$ occurs close to the beginning of $L$.

The changed search strategy has an advantage whenever short codewords occur much more often than longer ones. 
Fortunately, we have such a distribution, which is expressed formally in the following lemma:

\begin{lemma} \label{thmRareCharOccs}
The number of occurrences of all characters in the text whose associated Huffman codewords have lengths exceeding $2\log_\phi \sigma$ is $\Oh{n/\sigma}$.
\end{lemma}

\begin{proof}
For each character $c \in \Sigma$, 
let $f_c$ denote the number of occurrences of $c$ in the text 
and $d_c$ denote the depth of its associated leaf in the Huffman tree, which is also the length of its associated codeword.
Here we are interested in those characters $c$ for which $d_c > 2 \log_\phi \sigma = \log_\phi (\sigma^2)$. 
Since $d_c = \Oh{\log_\phi(n/f_c)}$~\cite[Thm~1]{katona76huffman} and
$\log_{\phi}$ is a strictly increasing function, we have $n / f_c = \Omega(\sigma^2)$, 
or, equivalently, $f_c = \Oh{n / \sigma^2}$.
The number of characters we deal with is upper-bounded by $\sigma$, therefore the total number of their occurrences is $\Oh{n / \sigma}$, which ends the proof. 
\end{proof}

For the \Oh{n/\sigma} characters specified in 
Lemma~\ref{thmRareCharOccs}
the exponential search works in 
\Oh{\min\{\lg\sigma, \lg\lg n\}} time.
However, for the remaining \Ot{n} characters
of the text the exponential search works in only \Oh{\lg(2\log_\phi \sigma)} = \Oh{\lg\lg\sigma} time.
Overall, the average time is \Oh{\lg\lg\sigma}, which is also the total average character decoding time, as finding the lexicographically smallest codeword with the appropriate length is, in general, more costly than all other required operations, which work in constant time.

\section{Multiple Fusion Trees} \label{sec:results}

In what follows, we present our space efficient representation of \Fusion{} and \First{} for achieving the result claimed in \cref{thmOverallResult}.
In \cref{secSuffixLength}, we start with the key observation that long codewords have a necessarily long prefix of ones.
This observation lets us group together codewords of roughly the same lengths, storing only the suffixes that distinguish them from each other.
Hence, we proceed with partitioning the set of codewords into codewords of different lengths, and orthogonal to that, of different distinguishing suffix lengths.

\subsection{Distinguishing Suffixes}\label{secSuffixLength} 
Let us notice at the beginning that a long codeword has a necessarily long prefix of ones: 

\begin{lemma}\label{thmCodewordSuffixes}
Let $\Codewords = \{C_1, \ldots, C_{\sigma}\}$ be a canonical code. 
Each codeword $C_i \in \Codewords$ can be represented as a binary sequence $C_i = \mathtt{1}^j b_{j+1} b_{j+2} \ldots b_{|C_i|}$ 
for some $j \geq 0$,
where 
$b_k \in \{0,1\}$ for all $k \in [j+1..|C_i|]$, 
and 
either (a) $|C_i| = j$, or (b) $b_{j+1} = 0$ and $j \ge |C_i| - \lg \sigma$. 
\end{lemma}
\begin{proof}

   \newcommand{\Nleft}{\ensuremath{v_{\mathup{l}}}}
\newcommand{\Nright}{\ensuremath{v_{\mathup{r}}}}
Without loss of generality, let $\sigma$ be a power of two, i.e., $\lg \sigma$ is an integer.
Let us assume there is a codeword~$C_i = \mathtt{1}^j b_{j+1} b_{j+2} \ldots b_{|C_i|}$ with $b_{j+1} = 0$. The length of the suffix~$S = b_{j+2} \ldots b_{|C_i|}$ of $C_i$ is $s := |S| = |C_i| - j - 1$. 
The prefix $\mathtt{1}^j b_{j+1}$ of $C_i$ is the string label of a node $\Nleft$ in the canonical tree.
The height of $\Nleft$ is at least $s$ since it has a leaf whose string label is $C_i$.
Since a canonical tree is a full binary tree, $\Nleft$ has a right sibling, which we call $\Nright$.
Since the depths of the leaves iterated in the left-to-right order in a canonical tree are non-decreasing,
all leaves of the tree rooted at $\Nright$ must be at depth at least $s$.
This implies that the number of leaves in the tree rooted in $\Nright$ is at least $2^{s}$.
There are at least two leaves in the tree rooted in $\Nleft$.
Moreover, the two trees rooted in $\Nleft$ and $\Nright$, respectively, are disjoint.
Consequently, there are at least $2^s + 2$ leaves in the tree representing $\Codewords$.
For $s \ge \lg \sigma$, we obtain a contradiction since the code tree has exactly $\sigma$ leaves.
Hence, $|C_i| - j - 1 = s < \lg \sigma$, i.e., $j \ge |C_i| - \lg \sigma$.
\end{proof}

The lemma states that given a codeword perceived as a concatenation of a maximal run of set \bsq{$\mathtt{1}$} bits and a suffix that follows it, 
the length of this suffix is less than $\lg \sigma$.
Given a length~$\ell$ and the lexicographically smallest codeword with length~$\ell$, and assume that it is of Case~(b) described in \cref{thmCodewordSuffixes}, 
then the proof of \cref{thmCodewordSuffixes} states we can charge it for $\Oh{2^s}$ nodes in the following argumentation:
Since the number of nodes is $2\sigma -1$ in a code tree,
we have that the number of codewords with a suffix of Case~(b) with a length $s > s'$ for a fixed $s'$
is at most $\Oh{\sigma/2^{s'}}$ by the pigeonhole principle.
We formalize this observation as follows:

\newcommand*{\BinaryAlphabet}{\ensuremath{\{\mathtt{0},\mathtt{1}\}}}

\begin{corollary}\label{corCodewordSuffix}
Given a length~$s$, the number of codewords~$C_i$ of a canonical code
given by $C_i = \mathtt{1}^p \mathtt{0} S$  with $S \in \BinaryAlphabet^{s-1}$ and $s \ge |C_i| - p$
is bounded by $\Oh{\sigma / 2^s}$.
\end{corollary}

In what follows, we want to derive a partition of the codewords stored in \Fusion{}.
Let us recall that \Fusion{}
stores not all codewords, but the codeword of the \emph{leftmost} leaf on each depth (and omits depth~$d$ if there is no leaf with depth~$d$).
Given a set of such codewords $C_1, \ldots, C_{\maxLength'}$ with $\maxLength' \le \maxLength$,
let $s_i$ denote the length of the shortest suffix of the representation $C_i = \mathtt{1}^p S$ with $S \in \BinaryAlphabet^{s_i}$.
In what follows, we call a codeword~$C_i$ \emph{long-tailed} if $s_i \ge 2 \lg \lg \sigma = \lg \lg^2 \sigma$, and otherwise \emph{short-tailed}.
A consequence of \cref{corCodewordSuffix} is that $\Oh{\sigma / \lg^2 \sigma}$ codewords can be long-tailed.
In what follows, we manage long-tailed and short-tailed codewords separately.

\subsection{Long-Tailed Codewords}\label{secLongTailed}
\newcommand*{\TinyFusion}[1]{\ensuremath{\mathsf{F}_{#1}}}

We partition the long-tailed codewords into sets $\Codewords_1, \ldots, \Codewords_m$ with $\Codewords_k$ being the set of codewords of lengths within the range $[1+(k-1)\lg \sigma, k\lg \sigma]$, for each $k \in [1..m]$.
By \cref{thmCodewordSuffixes}, we know that codewords in $\Codewords_k$ have the shape $PS$ with $P \in \{\mathtt{1}\}^*$, $|P| > (k-1) \lg \sigma$ and $|S| \le \lg \sigma$.
This means that we can represent a codeword of $\Codewords_k$ by its suffix~$S$ of length at most $\lg \sigma$.
Instead of representing \Fusion{} with a single fusion tree storing the lexicographically smallest codeword of~$\Codewords$ for each depth,
we maintain for each such set $\Codewords_k$ a dedicated fusion tree~\TinyFusion{k}.
To know which fusion tree to consult during a query with a bit string~$B$ read from a binary stream, we compute the longest prefix~$P \in \{\mathtt{1}\}^*$ of~$B$ (that is, a maximal run of ones).
We can do that by asking for the position most significant set bit of $B$ bit-wise negated,
which can be computed in constant time~\cite{fredman93fusion}.
Given that the most significant set bit of $B$ is at position $h \ge 1$ (counting from the leftmost position of $B$, which is $1$),
we know that $|P| = h-1$,
and consult $\TinyFusion{\upgauss{|P|/\lg \sigma}}$.

This works fine unless we encounter the following two border cases:
Firstly, $P$ is empty if and only if $h=1$, which happens when the next codeword is the lexicographically smallest codeword~$C_1$ with $C_1 = \texttt{0}^{|C_1|}$.
In that case, we know already that the answer is $C_1$ and are finished, so we treat $C_1$ individually. 
Secondly, the predecessor is not stored in the queried fusion tree, but in one built on shorter codewords.
To treat that problem, for $k \in [2..m]$, we store the dummy codeword $\mathtt{0}$ in \TinyFusion{k}
and cache the largest value of $\TinyFusion{1},\ldots,\TinyFusion{k-1}$ in \TinyFusion{k} such that \TinyFusion{k} returns this cached value instead of $\mathtt{0}$ 
in the case that $\mathtt{0}$ is the returned predecessor. 
Finally, we treat the border case for \TinyFusion{1}, in which we store the smallest codeword $C_1 = \First[1] = \mathtt{0}^{|C_1|}$ regardless of whether $C_1$ is long-tailed or not (in that way we are sure that a predecessor always exists), but store this codeword with the same amount of bits as the other codewords trimmed to their long-tailed suffix of $\lg \sigma$ bits.

Since the lengths of the codewords (before truncation) in the sets  $\Codewords_1, \ldots, \Codewords_m$ are pairwise disjoint, 
the number of elements stored in the fusion trees $\TinyFusion{1}, \ldots, \TinyFusion{m}$ is the number of long-tailed codewords, which is bounded by \Oh{\sigma/ \lg^2 \sigma}.
A key in a fusion tree represents a long-tailed codeword~$C$ by a pair consisting of $C$'s suffix of length $\lg \sigma$ (thus using $\lg \sigma$ bits)
and $C$'s lengths packed in $\lg \maxLength \le \lg \sigma$ bits.
Our total space for the long-tailed codewords is \Oh{\sigma / \lg \sigma} bits,
which are stored in \Oh{\maxLength / \lg \sigma} fusion trees.

\subsection{Short-Tailed Codewords}
Similar to our strategy for long-tailed codewords, we partition the short-tailed codewords by their total lengths.
By doing so, we obtain similarly a set of codewords
$\Codewords'_1, \ldots, \Codewords'_{m'}$ with $\Codewords'_k$ being the set of codewords of lengths within the range $[1+(k-1)\lg\lg^2 \sigma, k\lg\lg^2 \sigma]$. 
Like in \cref{secLongTailed}, we build a fusion tree on each of the codeword sets~$\Codewords'_k$ with the same logic.

But this time we know that each codeword~$C_i$ has the shape $\mathtt{1}^{|C_i|-s_i} B$ for $B \in \{0,1\}^{s_i}$ and $s_i < 2 \lg \lg \sigma$.
Additionally, we represent the length of the codeword~$C_i$ in the set $\Codewords'_k$ by the local depth within $\Codewords'_k$ as a $\Oh{\lg \lg \lg \sigma}$-bit integer.
Hence, the key of a fusion tree is composed by the $(2 \lg \lg \sigma)$-bit long suffix and the local depth with $\Oh{\lg \lg \lg \sigma}$ bits.
Therefore, the total space for the short-tailed codewords is $\Oh{\maxLength \lg \lg \sigma}$ bits,
which are stored in $\Oh{\maxLength / \lg \lg \sigma}$ fusion trees.

\subsection{Complexities}

Summing up the space for the short- and long-tailed codewords, we 
obtain $\Oh{\maxLength \lg \lg \sigma + \sigma / \lg \sigma}$ bits,
which are stored in $\Oh{\maxLength / \lg \lg \sigma}$ fusion trees.
To maintain these fusion trees in small space, 
we assume that we have access to a separately allocated RAM of above stated size to fit in all the data.
While we have a global pointer of $\Oh{w}$ bits to point into this space, 
pointers inside this space take $\Oh{\lg \sigma}$ bits.
Hence, we can maintain all fusion trees inside this allocated RAM with an extra space of $\Oh{\maxLength \lg \sigma / \lg \lg \sigma}$ bits.
In total, the space for $\Fusion$ is  $\Oh{\maxLength \lg \sigma / \lg \lg \sigma + \sigma / \lg \sigma} = \oh{\maxLength \lg\sigma + \sigma}$ bits.

To answer $\Fusion.\predecessor(B)$, we find the predecessor of $B$ among the short- and long-tailed codewords, and take the maximum one.
In detail, this is done as follows:
Remembering the decoding process outlined in \cref{secQueries}, $B$ is a bit string of length~$\maxLength$,
which we delegate to the fusion tree~$\TinyFusion{k}$ corresponding to either $\Codewords_k$ or $\Codewords'_k$ (we try both of them) if 
the longest unary prefix of \bsq{\texttt{1}}s in $B$ is in $[1+(k-1)\lg \sigma, k\lg \sigma]$ (long-tailed) or $[1+(k-1)\lg\lg^2 \sigma, k\lg\lg^2 \sigma]$ (short-tailed).
By delegation we mean that we remove this unary prefix, take the $\lg \sigma$ most significant bits (long-tailed) or the $2 \lg \lg \sigma$ most significant bits (short-tailed) from $B$, and stored these bits in a variable $B'$ used as the argument for the query $\TinyFusion{k}.\predecessor(B')$.

Finally, it remains to treat \First{} for encoding a character (cf.~\cref{secQueries}).
One way would be to partition \First{} analogously like \Fusion{}. 
Here, we present a solution based on the already analyzed bounds for \Fusion{}:
We create a duplicate of \Fusion{}, named \Fusion{}', whose difference to \Fusion{} is that the components of the stored pairs are swapped,
such that a predecessor query is of the form $\Fusion'.\predecessor((\ell, B))$ for a length~$\ell$ and a bit string $B$ of length $\maxLength$.
Then $\Fusion'.\predecessor((\ell, \mathtt{1}^{\maxLength} )) = (\ell', B')$, and
the $\ell$ first/leftmost bits of $B'$ are equal to $\First[\ell]$ if $\ell = \ell'$ (otherwise, if $\ell \not= \ell'$, then there is no leaf at depth~$\ell$).

\section{Conclusion and Open Problems}

Canonical codes (e.g., Huffman codes) are an interesting subclass of general prefix-free compression codes, allowing for compact representation without sacrificing character encoding and decoding times.
In this work, we refined the solution presented by Gagie et al.~\cite{gagie15efficient} and showed how to represent a canonical code, for an alphabet of size $\sigma$ and the maximum codeword length of $\maxLength$, in $\sigma \lg \maxLength (1 + \oh{1})$ bits, capable to encode or decode a symbol in constant worst case time. 
Our main idea was to store codewords not in their plain form, but partition them by their lengths and by the length of the shortest suffix covering all \bsq{\texttt{0}} bits of a codeword such that we can discard unary prefixes of \bsq{\texttt{1}}s from all codewords.

This research spawns the following open problems:
First, we wonder whether the proposed data structure works in the $AC^0$ model.
As far as we are aware of, the fusion tree can be modeled in the $AC^0$ model~\cite{andersson99fusion}, and our enhancements do not involve complicated operations except finding the most significant set bit, which can be also computed in $AC^0$~\cite{benkiki14towards}.
The missing part is the wavelet tree, for which we used the implementation of \citet[Thm~4.1]{belazzougui15wavelet}. 
We think that most bit operations used by this implementation are portable, and the used multiplications are not of general type, but a broadword broadcast operation, i.e., storing $\gauss{w/b}$ copies of a bit string of length~$b$ in a machine word of $w$ bits.
However, they rely on a monotone minimum perfect hash function, for which we do not know whether there is an existing alternative solution (even allowing a slight increase in the space complexity but still within $\sigma\lg\maxLength (1 + \oh{1})$ bits) in $AC^0$.

Another open question concerns the possibility of applying the presented technique in an adaptive (sometimes also called dynamic) Huffman coding scheme.
There exist efficient dynamic fusion tree and multi-ary wavelet tree implementations,
but it is unclear to us whether we can meet those costs (at least in the amortized sense) involved in maintaining the code tree dynamically.
Also, we are curious whether, for a small enough alphabet, a canonical code can encode and decode $k$ ($k > 1$) symbols at a time (ideally, $k = \Ot{w / \maxLength}$ or $k = \Ot{\log n / \maxLength}$, where the denominator could be changed to $\log\sigma$, if focusing on the average case) without a major increase in the required space.

\paragraph{Acknowledgments}
This work is funded by the JSPS KAKENHI Grant Number \texttt{JP21K17701}.

\bibliographystyle{abbrvnat}

\begin{thebibliography}{29}
\providecommand{\natexlab}[1]{#1}
\providecommand{\url}[1]{\texttt{#1}}
\expandafter\ifx\csname urlstyle\endcsname\relax
  \providecommand{\doi}[1]{doi: #1}\else
  \providecommand{\doi}{doi: \begingroup \urlstyle{rm}\Url}\fi

\bibitem[Aggarwal and Narayan(2000)]{aggarwal00efficient}
M.~Aggarwal and A.~Narayan.
\newblock Efficient {Huffman} decoding.
\newblock In \emph{Proc.\ ICIP}, pages 936--939, 2000.

\bibitem[Andersson et~al.(1999)Andersson, Miltersen, and
  Thorup]{andersson99fusion}
A.~Andersson, P.~B. Miltersen, and M.~Thorup.
\newblock Fusion trees can be implemented with {AC}\({}^{0}\) instructions
  only.
\newblock \emph{Theor. Comput. Sci.}, 215\penalty0 (1-2):\penalty0 337--344,
  1999.
\newblock \doi{10.1016/S0304-3975(98)00172-8}.

\bibitem[Belazzougui and Navarro(2015)]{belazzougui15wavelet}
D.~Belazzougui and G.~Navarro.
\newblock Optimal lower and upper bounds for representing sequences.
\newblock \emph{{ACM} Trans. Algorithms}, 11\penalty0 (4):\penalty0
  31:1--31:21, 2015.

\bibitem[Ben{-}Kiki et~al.(2014)Ben{-}Kiki, Bille, Breslauer, Gasieniec,
  Grossi, and Weimann]{benkiki14towards}
O.~Ben{-}Kiki, P.~Bille, D.~Breslauer, L.~Gasieniec, R.~Grossi, and O.~Weimann.
\newblock Towards optimal packed string matching.
\newblock \emph{Theor. Comput. Sci.}, 525:\penalty0 111--129, 2014.
\newblock \doi{10.1016/j.tcs.2013.06.013}.

\bibitem[Bentley and Yao(1976)]{BY76}
J.~L. Bentley and A.~C. Yao.
\newblock An almost optimal algorithm for unbounded searching.
\newblock \emph{Inf. Process. Lett.}, 5\penalty0 (3):\penalty0 82--87, 1976.
\newblock \doi{10.1016/0020-0190(76)90071-5}.

\bibitem[Buro(1993)]{buro93maximum}
M.~Buro.
\newblock On the maximum length of {Huffman} codes.
\newblock \emph{Inf. Process. Lett.}, 45\penalty0 (5):\penalty0 219--223, 1993.
\newblock \doi{10.1016/0020-0190(93)90207-P}.

\bibitem[Chowdhury et~al.(2002)Chowdhury, Kaykobad, and
  King]{chowdhury02efficient}
R.~A. Chowdhury, M.~Kaykobad, and I.~King.
\newblock An efficient decoding technique for {Huffman} codes.
\newblock \emph{Inf. Process. Lett.}, 81\penalty0 (6):\penalty0 305--308, 2002.

\bibitem[Claude et~al.(2015)Claude, Navarro, and Ord{\'o}{\~n}ez]{CNOis14}
F.~Claude, G.~Navarro, and A.~Ord{\'o}{\~n}ez.
\newblock The wavelet matrix: An efficient wavelet tree for large alphabets.
\newblock \emph{Information Systems}, 47:\penalty0 15--32, 2015.

\bibitem[Elias(1975)]{elias75universal}
P.~Elias.
\newblock Universal codeword sets and representations of the integers.
\newblock \emph{{IEEE} Trans. Inf. Theory}, 21\penalty0 (2):\penalty0 194--203,
  1975.

\bibitem[Fari{\~n}a et~al.(2016)Fari{\~n}a, Gagie, Manzini, Navarro, and
  Ord{\'o\~n}ez]{FGMNOspire16}
A.~Fari{\~n}a, T.~Gagie, G.~Manzini, G.~Navarro, and A.~Ord{\'o\~n}ez.
\newblock Efficient and compact representations of some non-canonical
  prefix-free codes.
\newblock In \emph{Proc.\ SPIRE}, LNCS 9954, pages 50--60, 2016.

\bibitem[Fari{\~n}a et~al.(2021)Fari{\~n}a, Gagie, Grabowski, Manzini, Navarro,
  and Ord{\'o\~n}ez]{farina2021efficient}
A.~Fari{\~n}a, T.~Gagie, S.~Grabowski, G.~Manzini, G.~Navarro, and
  A.~Ord{\'o\~n}ez.
\newblock Efficient and compact representations of some non-canonical
  prefix-free codes, 2021.
\newblock arXiv:1605.06615.

\bibitem[Fredman and Willard(1993)]{fredman93fusion}
M.~L. Fredman and D.~E. Willard.
\newblock Surpassing the information theoretic bound with fusion trees.
\newblock \emph{J. Comput. Syst. Sci.}, 47\penalty0 (3):\penalty0 424--436,
  1993.
\newblock \doi{10.1016/0022-0000(93)90040-4}.

\bibitem[Gagie et~al.(2015)Gagie, Navarro, Nekrich, and
  Pereira]{gagie15efficient}
T.~Gagie, G.~Navarro, Y.~Nekrich, and A.~O. Pereira.
\newblock Efficient and compact representations of prefix codes.
\newblock \emph{{IEEE} Trans. Inf. Theory}, 61\penalty0 (9):\penalty0
  4999--5011, 2015.

\bibitem[Hashemian(1995)]{hashemian95memory}
R.~Hashemian.
\newblock Memory efficient and high-speed search {Huffman} coding.
\newblock \emph{{IEEE} Trans. Commun.}, 43\penalty0 (10):\penalty0 2576--2581,
  1995.

\bibitem[Hirschberg and Lelewer(1990)]{hirschberg90efficient}
D.~S. Hirschberg and D.~A. Lelewer.
\newblock Efficient decoding of prefix codes.
\newblock \emph{Commun. {ACM}}, 33\penalty0 (4):\penalty0 449--459, 1990.

\bibitem[Huffman(1952)]{Huffman52}
D.~A. Huffman.
\newblock A method for the construction of minimum-redundancy codes.
\newblock \emph{Proceedings of the Institute of Radio Engineers}, 40\penalty0
  (9):\penalty0 1098--1101, 1952.

\bibitem[Katona and Nemetz(1976)]{katona76huffman}
G.~O.~H. Katona and T.~O.~H. Nemetz.
\newblock Huffman codes and self-information.
\newblock \emph{{IEEE} Trans. Inf. Theory}, 22\penalty0 (3):\penalty0 337--340,
  1976.
\newblock \doi{10.1109/TIT.1976.1055554}.

\bibitem[Klein(1997)]{klein97space}
S.~T. Klein.
\newblock Space- and time-efficient decoding with canonical {Huffman} trees.
\newblock In \emph{Proc.\ CPM}, volume 1264 of \emph{LNCS}, pages 65--75, 1997.

\bibitem[Klein(2000)]{klein00skeleton}
S.~T. Klein.
\newblock Skeleton trees for the efficient decoding of {Huffman} encoded texts.
\newblock \emph{Inf. Retr.}, 3\penalty0 (1):\penalty0 7--23, 2000.
\newblock \doi{10.1023/A:1009910017828}.

\bibitem[Kosolobov and Merkurev(2020)]{kosolobov20optimal}
D.~Kosolobov and O.~Merkurev.
\newblock Optimal skeleton {Huffman} trees revisited.
\newblock In \emph{Proc.\ CSR}, volume 12159 of \emph{LNCS}, pages 276--288,
  2020.

\bibitem[Liddell and Moffat(2006)]{liddell06decoding}
M.~Liddell and A.~Moffat.
\newblock Decoding prefix codes.
\newblock \emph{Softw. Pract. Exp.}, 36\penalty0 (15):\penalty0 1687--1710,
  2006.
\newblock \doi{10.1002/spe.741}.

\bibitem[Mansour(2007)]{mansour07efficient}
M.~F. Mansour.
\newblock Efficient {Huffman} decoding with table lookup.
\newblock In \emph{Proc.\ ICASSP}, pages 53--56, 2007.

\bibitem[Moffat(2019)]{Moffat19}
A.~Moffat.
\newblock Huffman coding.
\newblock \emph{{ACM} Comput. Surv.}, 52\penalty0 (4):\penalty0 85:1--85:35,
  2019.
\newblock \doi{10.1145/3342555}.

\bibitem[Moffat and Turpin(1997)]{moffat97implementation}
A.~Moffat and A.~Turpin.
\newblock On the implementation of minimum redundancy prefix codes.
\newblock \emph{{IEEE} Trans. Commun.}, 45\penalty0 (10):\penalty0 1200--1207,
  1997.

\bibitem[Navarro(2016)]{navarro16compact}
G.~Navarro.
\newblock \emph{Compact Data Structures -- A practical approach}.
\newblock Cambridge University Press, 2016.

\bibitem[Nekrich(2000)]{nekrich00decoding}
Y.~Nekrich.
\newblock Decoding of canonical {Huffman} codes with look-up tables.
\newblock In \emph{Proc.\ DCC}, page 566. {IEEE} Computer Society, 2000.
\newblock \doi{10.1109/DCC.2000.838213}.

\bibitem[P\v{a}tra\c{s}cu and Thorup(2014)]{patrascu14dynamic}
M.~P\v{a}tra\c{s}cu and M.~Thorup.
\newblock Dynamic integer sets with optimal rank, select, and predecessor
  search.
\newblock In \emph{Proc.\ FOCS}, pages 166--175, 2014.

\bibitem[Schwartz and Kallick(1964)]{SK64}
E.~S. Schwartz and B.~Kallick.
\newblock Generating a canonical prefix encoding.
\newblock \emph{Commun. ACM}, 7\penalty0 (3):\penalty0 166--169, 1964.
\newblock ISSN 0001-0782.
\newblock \doi{10.1145/363958.363991}.

\bibitem[Witten et~al.(1987)Witten, Neal, and Cleary]{witten87arithmetic}
I.~H. Witten, R.~M. Neal, and J.~G. Cleary.
\newblock Arithmetic coding for data compression.
\newblock \emph{Commun. {ACM}}, 30\penalty0 (6):\penalty0 520--540, 1987.
\newblock \doi{10.1145/214762.214771}.

\end{thebibliography}

\appendix
\section{Fusion Tree Augmentation}\label{secAugmentation}
Here, we provide an alternative solution to \cref{secMissingLeaves} without the need to let $\Fusion$ store pairs of codewords with their respective lengths.
Instead, the idea is to augment \Fusion{} with additional information to allow queries needed for simulating $\First$.
For that, we follow \citet[Section IV.]{patrascu14dynamic}, who presented a solution for a dynamic fusion tree that can answer the following queries, where $\mathcal{S}$ is the set of keys defined in \cref{secFormerApproach}.
\begin{itemize}
  \item $\rank(X)$ returns $| \{ Y \in \mathcal{S} : Y \le X \} |$ and
  \item $\select(i)$ returns $Y \in \mathcal{S}$ with $\rank(Y) = i$.
\end{itemize}
Since $\Fusion$ is a static data structure, adding these operations to $\Fusion$ is rather simple:
The idea is to augment each fusion tree node with an integer array storing the prefix-sums of subtree sizes.
Specifically, let $v$ be a fusion tree node having $w^c$ children, for a fixed (but initially selectable) positive constant $c < 1$.
Then we augment~$v$ with an integer array~$P_v$ of length $w^c$ such that $P_v[i]$ is the sum of the sizes of the subtrees rooted at the preceding siblings of $v$'s $i$-th child. 
With $P_v$ we can answer $\rank(x)$ via $\Fusion.\predecessor(x)$,
which is solved by traversing the fusion tree~$\Fusion$ in a top-down manner.
Starting with a counter~$c$ of the rank initially set to zero, 
on visiting the $i$-th child of a node~$v$, we increment $c$ by $P_v[i]$.
On finding $\predecessor(x)$ we return $c+1$ (+1 because the predecessor itself needs to be counted).

To answer a $\select$ query, we store the content of each $P_v$ in a (separate) fusion tree~$F_v$ to support a 
$\rank$ query on the prefix-sum values stored in $P_v$.
Consequently, a node of \Fusion{} stores not only $P_v$ but also a fusion tree built upon $P_v$.
The algorithm works again in a top-down manner, but uses $F_v$ for navigation:
For answering $\select(i)$, suppose we are at a node~$v$, 
and suppose that $F_v.\rank(i)$ gives us $j$.
Then we exchange $i$ with $i - P_v[j]$.
Now if $i$ has been decremented to zero, we are done since the answer is the $j$-th child of~$v$.
Otherwise, we descend to the $j$-th child of~$v$ and recurse.

The fusion trees $F_v$ as well as the prefix-sums $P_v$ take asymptotically the same space as its respective fusion tree node~$v$.\footnote{Since our tree~$\Fusion$ has a branching factor of $w^c$,
there are \Oh{m/w^c} leaves and \Oh{m/w^{2c}} internal nodes.
Since an internal node takes $w^{1+c}$ bits, we need $m/w^c$ bits for all internal nodes.
Also, a leaf~$\lambda$ does not need to store $F_\lambda$ and $P_\lambda$ since $P_\lambda[i] = i$.
}
Regarding the time, $F_v$ and $P_v$ answer a $\rank$ and an access query in $\Oh{\log_w w^c} = \Oh{1}$ time, 
and therefore, we can answer $\rank$ and $\select$ in the same time bounds as $\predecessor$.

It is left to deal with depths having no leaves.
For that, we add a bit vector \BV{D} of length~$\maxLength$ with $\BV{D}[\ell] = 1$ if and only if depth~$\ell$ has at least one leaf.
For the latter solution, $\Fusion$ only needs to take care of all depths in which leaves are present, i.e., 
$\Fusion$ will return a depth~$\ell'$ during a predecessor query, 
which we map to the correct depth with $\BV{D}.\select_1[\ell']$,
given that we endowed $\BV{D}$ with a select-support data structure in a precomputation step.
In total, \BV{D} with its select-support data structure takes $\maxLength + \oh{\maxLength}$ bits of space, and answers a select query in constant time. 
We no longer need to store $\First$ since $\First[\ell] = \Fusion.\select(\ell)$, given there is a leaf of the Huffman tree with depth $\ell$. 
Consequently, our approach shown in \cref{sec:results} works with this augmented fusion tree analogously.

\end{document}